\documentclass[12pt]{article}
\usepackage[utf8]{inputenc} % Standard writing
\usepackage[margin=1in]{geometry}
\usepackage[utf8]{inputenc}
\usepackage{amsmath,amssymb,amsfonts,amsthm}
\usepackage{bigints}
\usepackage{mathtools}
\usepackage[symbol]{footmisc}
\usepackage{footnote}
\usepackage{graphicx}
\usepackage[ruled,vlined]{algorithm2e}
\usepackage[colorinlistoftodos]{todonotes}
\usepackage{hyperref}
\theoremstyle{plain}
\newtheorem{theorem}{Theorem}
\theoremstyle{definition}

\newtheorem{corollary}[theorem]{Corollary}
\newtheorem{remark}{Remark}

\SetKwInput{KwInput}{Input}
\SetKwInput{KwOutput}{Output}
\begin{document}
%======================================================================
% Title
%======================================================================
% \begin{titlepage}
\title{Construction of integrable generalised travelling wave models and analytical solutions using Lie symmetries}
\author{Johannes Borgqvist$^{*\diamond}$, Fredrik Ohlsson$^{**}$, Xingjian Zhou$^{*}$, Ruth E. Baker$^{*}$}
\date{}
\maketitle
\begin{itemize}
\item[$*$] Wolfson Centre for Mathematical Biology, Mathematical Institute, University of Oxford,
\item[$**$] Department of Mathematics and Mathematical Statistics, Ume\aa$\,$University.
\item[$\diamond$] Corresponding author, e-mail: \url{johannes.borgqvist@gmail.com}.  
\end{itemize}
%\end{titlepage}
\rule{\textwidth}{1pt}

\begin{abstract}
Certain solutions of autonomous PDEs without any boundary conditions describing the spatiotemporal evolution of a dependent variable in an unbounded spatial domain can be characterised as a travelling wave moving with constant speed. In the simplest case, such PDEs can be reduced to a single autonomous second order ODE with one dependent variable. For certain parameter values it has been shown using perturbations methods in combination with ansätze that numerous such second order ODEs have analytical travelling wave solutions described by a simple sigmoid function. However, this methodology provides no leverage on the problem of finding a generalised class of models possessing such analytical travelling wave solutions. The most efficient methods for both finding analytical solutions and constructing classes of ODEs are based on Lie symmetries which are transformations known as one parameter $\mathcal{C}^{\infty}$ diffeomorphisms mapping solutions to other solutions. Recently, analytical solutions of a second order ODE encapsulating numerous oscillatory models as well as some of the previously mentioned travelling wave models with simple analytical solutions have been found by means of a two dimensional Lie algebra. Based on this Lie algebra, we construct the most general class of integrable autonomous second order ODEs for which these symmetries are manifest. Moreover, we show that a sub-class of second order ODEs has simple analytical travelling wave solutions described by a sigmoid function. Lastly, we characterise the action of the two symmetries in this Lie algebra on these simple analytical travelling wave solutions and we relate our sub-class of ODEs to previously known integrable travelling wave models.
\end{abstract}
\rule{\textwidth}{1pt}

%======================================================================
% Introduction
%======================================================================
\section{Introduction}
Many physical processes are described by an autonomous \textit{partial differential equation} (PDE) describing the spatiotemporal evolution of a dependent variable $u(x,t)$ depending on the two dependent variables $x$ and $t$ interpreted as space and time, respectively. A class of such PDEs referred to as reaction-diffusion-advection equations can be written on the form:
\begin{equation}
    \frac{\partial u}{\partial t}=\frac{\partial }{\partial x}\left(D(u)\frac{\partial u}{\partial x}\right)+\mu(u)\frac{\partial u}{\partial x}+g(u)\,,
    \label{eq:PDE_general}
\end{equation}
where the first term on the right hand side describes diffusion, the second term on the right hand side describes advection and the source term $g$ is referred to as the net-proliferation term. In the absence of boundary conditions, this PDE can be reduced to a second order \textit{ordinary differential equation} (ODE) by introducing a so called \textit{travelling wave variable} $z$ defined by $z=x-ct$ where the constant $c$ is referred to as the wave speed. The resulting ODE is given by
\begin{equation}
    0=\frac{\mathrm{d}}{\mathrm{d}z}\left(D(u)u_{z}\right)+(\mu(u)+c)u_{z}+g(u)\,,
    \label{eq:ODE_general}
\end{equation}
where $u(z)$ is the dependent variable, $z$ is the independent travelling wave variable and derivatives are denoted by a subscript, e.g. $\mathrm{d}u/\mathrm{d}z=u_z$. For the particular choices $D(u)=1$, $\mu(u)=0$ and $g(u)=u-u^{k}$ where $k$ is a constant, Kaliappan showed by means of perturbation expansions in combination with ans\"{a}tze that the resulting ODE in Eq. \eqref{eq:ODE_general} has analytical solutions given by~\cite{kaliappan1984exact}
\begin{equation}
    u(z)=\left(\frac{1}{1-\Upsilon\exp(mz)}\right)^{2/(k-1)}\,,
    \label{eq:kalliappan_solution}
\end{equation}
when the parameter constraints $c=(k+3)/\sqrt{2(k+1)}$, $m=(k-1)/\sqrt{2(k+1)}$ hold and where $\Upsilon$ is an arbitrary integration constant ~\cite{kaliappan1984exact}. When $k>1$ these analytical solutions are bounded and in this case they describe a travelling wave as they satisfy the boundary conditions $\underset{z\rightarrow\infty}{\lim}u(z)=0$ and $\underset{z\rightarrow-\infty}{\lim}u(z)=1$. In addition, the solutions are stable to all small finite domain disturbances and unstable to small disturbances in the far field~\cite{kaliappan1984exact}. A reason why these analytical solutions are interesting is due to the fact that the model analysed by Kaliappan in~\cite{kaliappan1984exact} is a generalisation of the Fisher-KPP model defined by $g(u)=u(1-u)$~\cite{fisher1937wave} which is used in numerous applications. Furthermore, an advantage with such simple analytical travelling wave solutions described by sigmoid functions as in Eq. \eqref{eq:kalliappan_solution}, is that they can readily be used to understand global properties of the overall dynamics of the system and what effect each parameter has on the dynamical behaviour. However, using perturbation methods it is difficult to determine when a second order ODE as in Eq. \eqref{eq:ODE_general} has simple analytical solutions such as the ones in Eq. \eqref{eq:kalliappan_solution}. 

Interestingly enough, using perturbation expansions it has been shown that another second order ODE also has simple analytical travelling wave solutions. In particular, the ODE in Eq. \eqref{eq:ODE_general} characterised by $D(u)=u^{-1}$, $\mu(u)=0$ and $g(u)=u(1-u)$ has analytical solutions given by $u(z)=(1+\Upsilon\exp(bz))^{-1}$ for an arbitrary integration constant $\Upsilon$ where $b=2(c+\sqrt{c^{2}+2})$~\cite{aronson1980density,newman1980some}. These results demonstrate that there are numerous second order ODEs having simple analytical travelling wave solutions given by sigmoid functions. Therefore, it is of interest to construct a class of second order ODEs characterised by analytical solutions with the same structure as in Eq. \eqref{eq:kalliappan_solution}, and this is a challenging problem to solve using perturbation methods. 

The most superior methods for finding analytical solutions and constructing classes of ODEs are based on \textit{Lie symmetries}, named after the Norwegian mathematician Sophus Lie. These transformations are so called (one parameter) $\mathcal{C}^{\infty}$ diffeomorphism which map a solution curve to another solution curve, or, equivalently formulated, transformations that leave the solution manifold invariant~\cite{bluman1989symmetries,hydon2000symmetry,olver2000applications,stephani1989differential}. This latter formulation is the key for constructing classes of models, since the implication of a symmetry leaving the solution manifold invariant is that any ODE for which certain symmetries are manifest can be written as a function of the so called \textit{differential invariants} of its symmetries. Thus, by calculating differential invariants of a set of symmetries one can construct the most general class of ODEs for which the symmetries of interest are manifest. To find analytical solutions using symmetries, one calculates coordinate transformations known as \textit{canonical coordinates} which transform the ODE of interest to an autonomous first order ODE which can be directly solved by means of integration, and this procedure is referred to as \textit{Lie's algorithm}~\cite{bluman1989symmetries}. In the particular case of a second order ODE, a set of two symmetries referred to as a two-dimensional \textit{Lie algebra} is required in order to find analytical solutions or first integrals, and such a Lie algebra yields two sets of canonical coordinates that can be used to carry out two successive step-wise integrations. A second order ODE with an associated two-dimensional Lie algebra is called \textit{integrable}, and for such ODEs one can calculate first integrals or analytical solutions. 

Previously, analytical solutions and first integrals of autonomous second order ODEs of Liénard-type have been found by means of a two dimensional Lie algebra based on a so called fibre-preserving symmetry~\cite{chandrasekar2006simple,feng2021symmetry,pandey2009lienard1,pandey2009lienard2}. Specifically, one such second order ODE encapsulates numerous oscillatory models, and in terms of the general ODE in Eq. \eqref{eq:ODE_general} it is defined by $D{(u)}=1$, $\mu(u)=k_{1}u^{q}+k_2$ and $g{(u)}=k_{3}u^{2q+1}+k_{4}u^{q+1}+\lambda_{1}u$~\cite{chandrasekar2006simple,feng2021symmetry} where $k_{1},k_{2},k_{3},k_{4},\lambda$ are arbitrary parameters and where $q$ is an arbitrary power. In particular, by choosing two of these parameters to $k_{1}=k_{4}=0$, we retrieve the same type of second order ODE corresponding to the previously mentioned generalised Fisher-KPP model analysed in~\cite{kaliappan1984exact} which has analytical travelling wave solutions given by Eq. \eqref{eq:kalliappan_solution}. In fact, it has been shown that the 
ubiquitous Fisher-KPP model is integrable under a two dimensional Lie algebra for the particular wave speed $c=5/\sqrt{6}$~\cite{güngör2021notes} which is exactly the wave speed for which it has an analytical solution~\cite{murray2002}. Accordingly, it should be possible to retrieve these analytical travelling wave solutions by means of Lie's algorithm based on the same type of two dimensional Lie algebra that underlie analytical solutions of the previously mentioned oscillatory models in~\cite{chandrasekar2006simple,feng2021symmetry,pandey2009lienard1,pandey2009lienard2}. Moreover, if such a two dimensional Lie algebra can be used to find analytical travelling wave solutions, it is of interest to quantify the action of these two symmetries on solution curves. 

In this work, we construct the most general second order ODE admitting the two-dimensional symmetry (sub-)algebra of certain classes of oscillatory second order ODEs in~\cite{chandrasekar2006simple} and density dependent diffusion models in~\cite{cherniha1997nonlineardiffusion}. Subsequently, we show that a sub-class of these models all have analytical solutions with the same structure as in Eq.~\eqref{eq:kalliappan_solution}. Finally, we describe the symmetry transformations generated by the vector fields in our Lie algebra and the action of these transformations on these analytical solutions.

%======================================================================
% Preliminaries
%======================================================================
\section{Preliminaries}
A Lie symmetry of a second order ODE $\Delta(z,u,u_z,u_{zz}) = 0$ in one dependent variable $u$ and one independent variable $z$ is a family of $\mathcal{C}^{\infty}$ diffeomorphisms $\Gamma_{\epsilon}$ parameterised by $\epsilon$ such that the transformations
\begin{equation}
\Gamma_{\epsilon}:(z,u)\mapsto(\hat{z}(\epsilon),\hat{u}(\epsilon))\,,
\label{eq:Gamma}
\end{equation}
map a solution curve $(z,u)$ to another solution curve $(\hat{z}(\epsilon),\hat{u}(\epsilon))$ and constitute a one parameter Lie group. Such a symmetry is completely characterised by its infinitesimal description in terms of the vector field
\begin{equation}
X=\xi{(z,u)}\partial_z+\eta{(z,u)}\partial_u\,,
\label{eq:X}
\end{equation}
known as the \textit{infinitesimal generator of the Lie group}~\cite{bluman1989symmetries,hydon2000symmetry,olver2000applications,stephani1989differential} where $\xi$ and $\eta$ are known as the \textit{infinitesimals}. The symmetries $\Gamma_{\epsilon}$ of the ODE $\Delta = 0$ can be found by solving the \textit{linearised symmetry condition}~\cite{bluman1989symmetries,hydon2000symmetry,olver2000applications,stephani1989differential}
\begin{equation}
X^{(2)}\left.(\Delta)\right|_{\Delta = 0}=0 \,.
\label{eq:lin_sym_general}
\end{equation}
In the linearised symmetry condition, the second prolongation of $X$ is given by
\begin{equation}
X^{(2)}=X+\eta^{(1)}(z,u,u_{z})\partial_{u_{z}}+\eta^{(2)}(z,u,u_{z},u_{zz})\partial_{u_{zz}}\,,
\label{eq:prolonged}
\end{equation}
with the \textit{prolonged infinitesimals} $\eta^{(1)}$ and $\eta^{(2)}$ given by~\cite{hydon2000symmetry}
  \begin{align}
    \eta^{(1)} &= \eta_z+(\eta_u-\xi_z)u_{z}-\xi_uu_{z}^2,\label{eq:eta_1}\\
    \eta^{(2)} &=\eta_{zz}+(2\eta_{zu}-\xi_{zz})u_{z}+\left(\eta_{uu}-2\xi_{zu}\right)u_{z}^2-\xi_{uu}u_{z}^3+\left(\eta_u - 2\xi_z -3\xi_u u_{z}\right)u_{zz}.\label{eq:eta_2}  \end{align}

% Canonical coordinates
\subsection{Canonical coordinates and differential invariants}
The \textit{canonical coordinates} of a symmetry is a set of coordinates $(s,r)=(s{(z,u)},r{(z,u)})$ in which the action of the symmetry is a translation in the independent coordinate~\cite{bluman1989symmetries,hydon2000symmetry,olver2000applications,stephani1989differential}, i.e.
\begin{equation}
Xs=\xi\frac{\partial s}{\partial z}+\eta\frac{\partial s}{\partial u}=1\,,\quad X r=\xi\frac{\partial r}{\partial z}+\eta\frac{\partial r}{\partial u}=0\,.\\
\label{eq:canonical}
\end{equation}
Following Lie's original construction, the canonical coordinates can be used to reduce the original equation by quadrature.

%These canonical are used to find analytical solutions in three steps in what is referred to as \textit{Lie's algorithm}~\cite{bluman1989symmetries}: (I) find the canonical coordinates by solving Eq. \eqref{eq:canonical} using the method of characteristics, (II) express the travelling wave RD model in Eq. \eqref{eq:general_RD} in these canonical coordinates and integrate the resulting ODE and (III) substitute back to the original coordinates.
%One of the main applications of these canonical coordinates is to find analytical solutions of the original ODE and the key for finding such solutions is to express the original ODE in terms of these canonical coordinates as the resulting equation becomes separable.

A further characterisation of the symmetry $\Gamma_{\epsilon}$ is obtained by constructing a complete set of  \textit{differential invariants}, i.e., non-constant functions $I=I(z,u,u_{z},u_{zz})$ satisfying
\begin{equation}
    X^{(2)}I=\xi\frac{\partial I}{\partial z}+\eta\frac{\partial I}{\partial u}+\eta^{(1)}\frac{\partial I}{\partial u_{z}}+\eta^{(2)}\frac{\partial I}{\partial u_{zz}}=0\,.
    \label{eq:differential_invariant}
\end{equation}
In the general case of single second order ODEs considered here, the one parameter group generated by $X$ has one zeroth order invariant corresponding to the canonical coordinate $r(z,u)$ in Eq.~\eqref{eq:canonical}, one first order invariant $v(z,u,u_{z})$, and one second order invariant $w(z,u,u_{z},u_{zz})$. The fact that the space of solutions is invariant under the action of the symmetry implies that the most general class of second order ODEs for which a symmetry generated by $X$ in Eq. \eqref{eq:X} is manifest can be written as a function of its invariants as
 \begin{equation}
    w(z,u,u_{z},u_{zz})=F\left(r(z,u),v(z,u,u_{z})\right)\,,
     \label{eq:general_ODE}
 \end{equation}
 where $F\in\mathcal{C}\left(\mathbb{R}^{2}\right)$ is an arbitrary function~\cite{hydon2000symmetry,olver2000applications}.  

% Lie Algebra
 \subsection{Lie algebras}
 A finite-dimensional set of infinitesimals $\mathcal{L}=\mathrm{Span}\{X_{1},\ldots,X_{n}\}$ for some finite $n\in\mathbb{N}_{+}$ constitutes a vector space. The bilinear map referred to as the \textit{Lie bracket} is defined by
 \begin{equation}
   [X_{i},X_{j}] = X_{i}X_{j}-X_{j}X_{i}\,,\quad i,j\in\{1,\ldots,n\}\,.
   \label{eq:Lie_bracket}
 \end{equation}
 If $\mathcal{L}$ is closed under the action of the Lie bracket it is referred to as a \textit{Lie algebra}, and in this case it follows that
 \begin{equation}
   [X_{i},X_{j}] = \sum_{k=1}^{n}c_{k}X_{k}\,,\quad i,j\in\{1,\ldots,n\}\,.
   \label{eq:Lie_bracket_2}
 \end{equation}
 A \textit{sub-algebra} is a sub-set of $\mathcal{L}$ that is closed under the action of the Lie bracket. 

% Lie Algebra
 \subsection{General class of autonomous models resulting from a fibre preserving symmetry}\label{ssec:fibre}
 A \textit{fibre preserving} symmetry is a symmetry where the changes in the independent variable do not depend on the dependent variable~\cite{güngör2021notes}. Such symmetries are given by 
 \begin{equation}
    X_{G}=\xi(z)\partial_z + f(z)u\partial_{u}\,,
     \label{eq:fibre_preserving}
 \end{equation}
 and the general classes of second order ODE for which these symmetries are manifest are given by~\cite{güngör2021notes}
 \begin{equation}
    u_{zz}+\xi^{-1}(\xi_{z}-2f)u_{x}+\xi^{-2}(f^{2}-\xi f_{z})u=\frac{H(r,v)}{\nu\xi^{2}}\,,\quad\xi\neq 0\,,
     \label{eq:class_fibre}
 \end{equation}
 where $H$ is an arbitrary function and where~\cite{güngör2021notes}
 \begin{equation}
    \nu=\exp\left(\int\frac{f}{\xi}\right)\,.
     \label{eq:nu}
 \end{equation}
Several classes of second order ODEs that are invariant under the action of fibre preserving symmetries have been constructed by G{\"u}ng{\"o}r in~\cite{güngör2021notes}. In the specific case of a two parameter symmetry group extended by the translation generator 
 \begin{equation}
  X_{z}=\partial_{z}\,,
  \label{eq:translation}
\end{equation}
 the commutation relation
 \begin{equation}
[X_{z},X_{G}]=pX_{G}\,,
\end{equation}
must hold in order for $\mathcal{L}=\mathrm{Span}\{X_{z},X_{G}\}$ to be a Lie algebra. This leads to an autonomous class of second order ODEs in Eq. \eqref{eq:class_fibre} and in this case, the functions $\xi(z)$ and $f(z)$ in Eq. \eqref{eq:fibre_preserving} must satisfy $\xi(z)=-K\exp(pz)$ and $f(z)=\exp(pz)$, respectively. Thus, the resulting generator in this case is given by 
\begin{equation}
X_{G}=-K\exp\left(pz\right)\partial_{z}+\exp\left(pz\right)u\partial_{u}\,,
\label{eq:generalised_generator}
\end{equation}
defined by the two arbitrary parameters $K,p$. In this case, $\mathcal{L}=\mathrm{Span}\{X_{G}\}$ is an \textit{ideal} of $\mathcal{L}$, and $\{X_{z},X_{G}\}$ is its canonical basis. Moreover, this generator $X_{G}$ in Eq. \eqref{eq:generalised_generator} encapsulates the generating vector fields of the symmetries of the general Liénard equation\footnote{The form of the generator $X_G = -K\exp\left(pz\right)\partial_{z}+\exp\left(pz\right)\left(u+\lambda\right)\partial_{u}$ sometimes appearing in~\cite{pandey2009lienard1} is recovered by shifting the state $u$ by a constant to $u+\lambda$.}\cite{pandey2009lienard1}.
%======================================================================
% Results
%======================================================================
\section{Results}
We present three main results. First, we present of a class of autonomous integrable second order ODEs based on a Lie algebra known to underlie analytical solutions of several travelling wave and oscillatory models considered in~\cite{chandrasekar2006simple,feng2021symmetry,güngör2021notes,pandey2009lienard1,pandey2009lienard2}. Second, we derive a sub-class of second order ODEs which has simple analytical travelling wave solutions on the same form as in Eq. \eqref{eq:analytical_solution}. Also, we demonstrate that the generalised Fisher--KPP model which Kaliappan derived analytical solutions of in~\cite{kaliappan1984exact} is indeed a member of this sub-class of second order ODEs. Third, we generate the two symmetries in our Lie algebra in order to quantify their action on the simple analytical travelling wave solution curves characterising our sub-class of second order ODEs. 
% -----------------------------------------------------------------------
% Construction of class of models using differential invariants
% -----------------------------------------------------------------------
\subsection{Construction of integrable second order ODEs using differential invariants}
We wish to construct a class of second order ODEs with simple analytical travelling wave solutions. Moreover, it is known that the generalisation of the Fisher--KPP model studied by Kaliappan in~\cite{kaliappan1984exact} has such analytical solutions which are given by Eq. \eqref{eq:kalliappan_solution}. Previously, analytical solutions and first integrals of a second order ODE encompassing numerous oscillatory models~\cite{chandrasekar2006simple,feng2021symmetry} have been found using canonical coordinates of a two-dimensional Lie algebra, and importantly Kaliappans generalisation is also encapsulated by this same second order ODE. Consequently, it is expected that the analytical solutions in Eq. \eqref{eq:kalliappan_solution} can be obtained by means of Lie's algorithm using canonical coordinates derived from the two-dimensional Lie algebra of the oscillatory models in~\cite{chandrasekar2006simple,feng2021symmetry}.

To this end, as a first step towards constructing a class of second order ODEs with simple analytical travelling wave solutions, we are interested in constructing an integrable class of autonomous second order ODEs based on the symmetry algebra of the general oscillatory models in~\cite{chandrasekar2006simple}. This is achieved using Gung\"{o}r's class of second order ODEs~\cite{güngör2021notes} based on the fibre preserving vector field $X_{G}$ in Eq. \eqref{eq:generalised_generator} in combination with the translation generator $X_{z}$ in Eq. \eqref{eq:translation} (see Section \ref{ssec:fibre} of Methods). The most general class of autonomous second order ODEs that is invariant under the action of the Lie algebra $\mathcal{L}=\mathrm{Span}\{X_{z},X_{G}\}$ is given by 
\begin{equation} u_{zz}+\frac{(2+Kp)}{K}u_{z}+\frac{(1+Kp)}{K^{2}}u=u^{(2Kp+1)}F_{zG}\left(\frac{1}{u^{Kp}}\left(\frac{u_{z}}{u}+\frac{1}{K}\right)\right)\,,
    \label{eq:general_class_autonomous}
  \end{equation}  
  where $F_{zG}\in\mathcal{C}(\mathbb{R})$ is an arbitrary function. The details behind these calculations are presented in Section 1 of the supplementary material.

We note that the quantity
\begin{equation}
    r_{1}(z,u,u_{z})=\frac{1}{u^{Kp}}\left(\frac{u_{z}}{u}+\frac{1}{K}\right)\,,
      \label{eq:r_1}
  \end{equation}
is a first order differential invariant of the generating vector field $X_G$ in Eq. \eqref{eq:generalised_generator}. Also, the second order ODE encompassing numerous oscillatory models in~\cite{chandrasekar2006simple} is recovered by the choice $F_{zG}(r_{1})=k_{1}+k_{2}r_{1}$ in Eq.~\eqref{eq:general_class_autonomous}. In addition, the class of autonomous models in Eq. \eqref{eq:general_class_autonomous} is integrable by virtue of their symmetry under a two-dimensional Lie algebra, and can consequently be integrated using quadrature following symmetry reduction using the algebra $\mathcal{L}$.

% -----------------------------------------------------------------------
% Class of models with analytical solutions
% -----------------------------------------------------------------------
\subsection{A sub-class of second order ODEs with analytical solutions}
Analytical solutions and first integrals of the general class of second order ODEs in Eq. \eqref{eq:general_class_autonomous} can be found by means of step-wise integration based on two sets of canonical coordinates. Starting with the sub-algebra $\mathcal{L}_{G}=\mathrm{Span}\{X_{G}\}$, the corresponding canonical coordinates are given by
\begin{equation}
    s_{2}{(z)}=\frac{1}{Kp}\exp(-pz)\,,\quad r_{2}{\left(z,u\right)}=u\exp\left(\frac{1}{K}z\right)\,.
    \label{eq:canonical_coordinates_2}
\end{equation}
Here, the sub-script 2 indicates that these canonical coordinates are used in the second and last step-wise integration. Importantly, the first order invariant $r_1$ in Eq. \eqref{eq:r_1} can be formulated in terms of the canonical coordinates $(s_{2},r_{2})$ as follows
  \begin{equation} 
  r_{1}=-\left(\frac{1}{Kr_{2}^{Kp+1}}\right)\frac{\mathrm{d}r_{2}}{\mathrm{d}s_{2}}\,.
    \label{eq:r_1_canonical}
  \end{equation}
The next set of canonical coordinates is derived by applying $X_{z}$ to the differential invariants $(r_{2},r_{1})$. In this case, the corresponding canonical coordinates are given by
\begin{equation}
  s_{1}(r_{2},r_{1})=K\ln(r_{2})\,,\quad r_{1}(r_{2},r_{1})=r_{1}\,,
  \label{eq:canonical_full_article}
\end{equation}  
where $r_{1}$ is the first order differential invariant in Eq. \eqref{eq:r_1}. Importantly, in terms of the canonical coordinates $(s_{1},r_{1})$, the class of second order ODEs in Eq. \eqref{eq:general_class_autonomous} can be formulated as the following autonomous first order ODE

\begin{equation}
K r_{1}\frac{\mathrm{d}r_{1}}{\mathrm{d}s_{1}}=F_{zG}(r_{1})-(Kp+1)r_{1}^{2}\,.
  \label{eq:ODE_general_class}
\end{equation}
Consequently, analytical solutions and first integrals of the class of second order ODEs in Eq. \eqref{eq:general_class_autonomous} are obtained by means of two successive step-wise integrations. First, the first order ODE in Eq. \eqref{eq:ODE_general_class} is solved for $r_{1}(s_{1})$. Second, the resulting solution for $r_{1}(s_{1})$ yields a first order autonomous ODE for $r_{2}$ by re-writing $s_{1}$ and $r_{1}$ in terms of $r_{2}$ using Eqs. \eqref{eq:canonical_full_article} and \eqref{eq:r_1_canonical}, respectively. In particular, analytical travelling wave solutions of the same form as in Eq. \eqref{eq:kalliappan_solution} correspond to solutions of Eq. \eqref{eq:ODE_general_class} where the first order differential invariant $r_{1}$ in Eq. \eqref{eq:r_1} is \textit{constant}, i.e. $r_{1}=\omega$ for some constant $\omega\neq 0$. The details behind the calculations of these canonical coordinates are presented in Section 2 of the supplementary material. 

Using these sets of canonical coordinates, we derive a sub-class of models sharing simple analytical solutions (Theorem \ref{thm:analytical_solution}).

\begin{theorem}[\textit{A sub-class of second order ODEs with simple analytical solutions}]
The sub-class of models defined by functions $F_{zG}$ in Eq. \eqref{eq:general_class_autonomous} such that the equation:
\begin{equation}
0=F_{zG}(\omega)-(Kp+1)\omega^{2}\,,
  \label{eq:param_class}
\end{equation}
has at least one non-zero solution $\omega\in\mathbb{R}:\omega\neq 0$ has analytical solutions given by:
\begin{equation}
    u(z)=\left(\frac{1}{(K\omega)^{1/Kp}}\right)\left(\frac{1}{\left(1+\Upsilon\exp(pz)\right)^{1/Kp}}\right)\,,
        \label{eq:analytical_solution}
    \end{equation}
where $\Upsilon$ is an arbitrary integration constant.
        \label{thm:analytical_solution}
\end{theorem}
\begin{proof}
    See Section 4 of the supplementary material.
\end{proof}
\begin{remark}
In the case when $\omega=0$ is a solution of Eq. \eqref{eq:param_class}, analytical solutions are given by $u(z)=\Upsilon\exp\left(-(1/K)z\right)$.
\end{remark}
To ensure that the analytical solutions in Eq. \eqref{eq:analytical_solution} characterise a travelling wave, we impose parameter conditions (Corollary \ref{thm:travelling_wave}).
\begin{corollary}[\textit{Positivity of the parameters ensures travelling waves}]
The analytical solutions characterising the sub-class of models in Theorem \ref{thm:analytical_solution} are described by a travelling wave if the parameters $K,p$ as well as the root $\omega$ are positive. Given the positivity of these parameters, the analytical solutions in Eq. \eqref{eq:analytical_solution} satisfy the following limits
\begin{equation}
    \underset{z\rightarrow-\infty}{\lim}u(z)=\frac{1}{(K\omega)^{1/Kp}}\,,\quad\underset{z\rightarrow\infty}{\lim}u(z)=0\,.
    \label{eq:limit_TW}
\end{equation}
\label{thm:travelling_wave}
\end{corollary}
\begin{remark}
The quantity $(K\omega)^{-1/Kp}$ corresponds to the carrying capacity of the travelling wave solution.
\end{remark}
Given Theorem \ref{thm:analytical_solution} and Corollary \ref{thm:travelling_wave}, we can now readily visualise our class of models giving rise to analytical travelling wave solutions. By considering Eq. \eqref{eq:param_class} under the assumptions that $K,p,\omega$ are all positive, we see that the class of models giving rise to analytical travelling wave solutions are determined by functions $F_{zG}(\omega)$ that intersect the monomial $(Kp+1)\omega^{2}$ for some positive $\omega\in\mathbb{R}_+$. For instance, this implies that all continuous functions $F_{zG}$ that are defined on the whole of $\mathbb{R}$ that satisfy $F_{zG}(\omega=0)>0$ has analytical travelling wave solutions given by Eq. \eqref{eq:analytical_solution}. The counter intuitive conclusion from this result is that numerous seemingly complicated second order ODEs in fact admits simple analytical travelling wave solutions given by Eq. \eqref{eq:analytical_solution}.

Next, we demonstrate that the previously mentioned generalisation of the Fisher-KPP model studied by Kaliappan~\cite{kaliappan1984exact} is indeed encapsulated by this class of models.

% -----------------------------------------------------------------------
% Reaction diffusion example
% -----------------------------------------------------------------------
\subsubsection{Analytical solutions of a generalisation of the Fisher-KPP model}
By choosing $F_{zG}(r_{1})=a$ in the general class of second order ODEs in Eq. \eqref{eq:general_class_autonomous} where $a\neq 0$ is a constant, we obtain the following second order ODE
\begin{equation}
u_{zz}+\frac{(2+Kp)}{K}u_{z}+\frac{(1+Kp)}{K^{2}}u
-a\,u^{(2Kp+1)}=0\,,
    \label{eq:RD_model}
\end{equation}
where we assume that $K,p>0$ in accordance with Corollary \ref{thm:travelling_wave}. This second order ODE encapsulates the generalisation of the ubiquitous Fisher-KPP model studied by Kaliappan having analytical solutions given by Eq. \eqref{eq:kalliappan_solution}. Moreover, in order for Eq. \eqref{eq:param_class} to have solutions $\omega$ whenever $K,p>0$ we impose that $a>0$, and in this case these roots are given by 
\begin{equation}
    \omega=\pm\sqrt{\frac{a}{Kp+1}}\,,
    \label{eq:roots_RD}
\end{equation}
which, in turn, yield analytical solutions given by Eq. \eqref{eq:analytical_solution}. In particular, travelling wave analytical solutions are given by the positive root in Eq. \eqref{eq:roots_RD} according to Corollary \ref{thm:travelling_wave}. 

As an example, we substitute the parameter values $p=1/\sqrt{6}$, $K=\sqrt{3/2}$, $Kp=1/2$ and $a=1$ in Eq. \eqref{eq:RD_model} which results in the Fisher-KPP model~\cite{fisher1937wave} characterised by the net proliferation term $g(u)=u(1-u)$ with the specific wave speed $c=5/\sqrt{6}$. In this case, the positive root in Eq. \eqref{eq:roots_RD} is given by $\omega=\sqrt{2/3}=1/K$ and the carrying capacity is given by $(K\omega)^{-1/Kp}=1$ according to Eq. \eqref{eq:limit_TW}. This is exactly the known analytical solution of the Fisher-KPP model~\cite{murray2002} for the wave speed $c=5/\sqrt{6}$.

In total, this work demonstrates that the analytical travelling wave solutions in Eq. \eqref{eq:analytical_solution}, can in fact be derived from the Lie algebra $\mathcal{L}=\mathrm{Span}\{X_{z},X_{G}\}$.  Given that Lie symmetries underlie these analytical solutions, it is of interest to understand how transformations by these two symmetries affect the curves described by the analytical solutions. To this end, we quantify the actions of the symmetries generated by the vector fields in the Lie algebra $\mathcal{L}=\mathrm{Span}\{X_{z},X_{G}\}$ on these solution curves. 
% -----------------------------------------------------------------------
% Quantifying action of symmetry 
% -----------------------------------------------------------------------
\subsection{Quantifying the action of symmetries on analytical solutions}
 The vector field $X_{z}$ in Eq. \eqref{eq:translation} generates the symmetry $\Gamma_{\epsilon}^{z}:(z,u)\mapsto(z+\epsilon,u)$ corresponding to translations in $z$. Moreover, the vector field $X_{G}$ in Eq. \eqref{eq:generalised_generator} generates the symmetry $\Gamma_{\epsilon}^{G}$ given by
\begin{equation}
\Gamma_{\epsilon}^{G}:(z,u)\mapsto\left(\frac{1}{p}\ln\left(\frac{\exp(pz)}{1+pK\epsilon\exp(pz)}\right),u(1+pK\epsilon\exp(pz))^{1/pK}\right)\,.
  \label{eq:generalised_symmetry}
\end{equation}
For the analytical solutions in Eq. \eqref{eq:analytical_solution}, we quantify the action of these symmetries in terms of the arbitrary parameter $\Upsilon$. The action of the translation symmetry corresponds to
\begin{equation}
  \Gamma_{\epsilon}^{z}:\Upsilon\mapsto\Upsilon\exp\left(-p\epsilon\right)\,,
  \label{eq:Gamma_z_param}
\end{equation}
and, similarly, the action of the generalised symmetry is given by
\begin{equation}
  \Gamma_{\epsilon}^{G}:\Upsilon\mapsto\left(\Upsilon-pK\epsilon\right)\,.
  \label{eq:Gamma_G_param}
\end{equation}
These results imply that for both these symmetries both the carrying capacity in Eq. \eqref{eq:limit_TW} as well as the power of the exponential term given by $1/Kp$ in the solution curves in Eq. \eqref{eq:analytical_solution} are invariant under transformations by these symmetries. The details behind these calculations are found in Section 5 of the supplementary material.
%======================================================================
% Discussion
%======================================================================
\section{Discussion}
In this work, we construct the most general class of integrable second order ODEs based on their symmetry under a generalised Lie algebra common to previously considered travelling wave models~\cite{chandrasekar2006simple,cherniha1997nonlineardiffusion,feng2021symmetry}.
We show that a subset of these models admit analytical travelling wave solutions of a simple form, previously known to exist for distinct travelling wave models~\cite{aronson1980density,newman1980some,kaliappan1984exact}, and derive the condition on the function $F_{zG}$ defining this subset. The analytical solutions are found by means of integration of the original equation reduced with respect to its symmetries and, more specifically, as solutions invariant under the action of the generalised symmetry $X_G$. The general class of models include several of the previously mentioned travelling wave models such as the Fisher-KPP model as special cases, but constitutes a much larger set of second order ODEs with fewer restrictions on the dynamics. The striking and counter intuitive conclusion from our class of travelling wave models is that we can construct seemingly complicated second order ODEs, which all share simple analytical travelling wave solutions with the ubiquitous Fisher-KKP model. From the point of view of mechanistic modelling where models are often constructed based on physical assumptions that can be hard to validate, this work illustrates how models instead can be constructed based on a mathematical principle in the form of Lie symmetries. 

%======================================================================
% Acknowledements
%======================================================================
\section{Acknowledgements}
JGB would like to thank the Wenner--Gren Foundations for a research fellowship, and Linacre College of the University of Oxford for a Junior Research Fellowship. JGB and XZ would like to thank the London Mathematical Society for the Undergraduate Research Bursary with grant number URB-2023-46 which funded an 8-week summer project. 

%======================================================================
% References
%======================================================================

\setcounter{section}{0}

\clearpage

\noindent \textit{\Large Construction of integrable generalised travelling wave models and analytical solutions using Lie symmetries\\[0.5cm] Supplementary material\\}\normalsize\\[2mm]
\author{Johannes Borgqvist$^{*\diamond}$, Fredrik Ohlsson$^{**}$, Xingjian Zhou$^{*}$, Ruth E. Baker$^{*}$}
\date{\today}
\maketitle
\begin{itemize}
\item[$*$] Wolfson Centre for Mathematical Biology, Mathematical Institute, University of Oxford,
\item[$**$] Department of Mathematics and Mathematical Statistics, Ume\aa$\,$University.
\item[$\diamond$] Corresponding author, e-mail: \url{johannes.borgqvist@gmail.com}.  
\end{itemize}
%\end{titlepage}
\rule{\textwidth}{1pt}

\pagenumbering{arabic}

%\tableofcontents
%======================================================================
% Proof of Theorem 1
%======================================================================
\section{Model construction using differential invariants}

We consider the infinitesimal generator of the generalised symmetry $X_{G}$ in Eq. \eqref{eq:generalised_generator}. Given this generator, the infinitesimals are given by
\begin{align}
\xi{(z,u)}&=- K\exp\left(pz\right)\,,\label{eq:xi_generalised}\\
\eta{(z,u)}&=u \exp\left(pz\right)\,.\label{eq:eta_generalised}
\end{align}
The corresponding prolonged infinitesimals are given by
\begin{align}
\eta^{(1)}{(z,u,u_{z})}&=\left(p u + u_{z} \left(Kp + 1\right)\right) \exp\left(pz\right)\,,\label{eq:eta_1_general}\\
\eta^{(2)}{(z,u,u_{z},u_{zz})}&=\left(p^{2} u + p u_{z} \left(Kp + 2\right) + u_{zz} \left(2 Kp + 1\right)\right) \exp\left(pz\right)\,.\label{eq:eta_2_general}
\end{align}
Given these infinitesimals, we calculate the differential invariants. The zeroth order invariant is a \textit{first integral} of
\begin{equation}
  \dfrac{\mathrm{d}u}{\mathrm{d}z}=-\frac{1}{K}u\,,
  \label{eq:r_eq}
\end{equation}
which can be formulated as
\begin{equation}
  r{\left(z,u\right)}=u\exp\left(\frac{1}{K}z\right)\,.
  \label{eq:r_sol}
\end{equation}
The first order invariant is a \textit{first integral} of:
\begin{equation}
  \dfrac{\mathrm{d}u_{z}}{\mathrm{d}u}=p+(1+Kp)\frac{u_{z}}{u}\,,
  \label{eq:v_eq}
\end{equation}
which can be formulated as follows:
\begin{equation} v{\left(z,u,u_{z}\right)}=\frac{1}{u^{Kp}}\left(\frac{u_{z}}{u}+\frac{1}{K}\right)\,.
  \label{eq:v_sol}
\end{equation}
The second order invariant is a \textit{first integral} of 
\begin{equation} \dfrac{\mathrm{d}u_{zz}}{\mathrm{d}u}=p^{2}+p(2+Kp)\frac{u_{z}}{u}+(2Kp+1)\frac{u_{zz}}{u}\,.
  \label{eq:w_eq}
\end{equation}
 From Eq. \eqref{eq:v_sol}, it follows that
\begin{equation}
  \frac{u_{z}}{u}=vu^{Kp}-\frac{1}{K}\,,
  \label{eq:v_subs}
\end{equation}
where $v=v{\left(z,u,u_{z}\right)}$ is the first order invariant. Using Eq. \eqref{eq:v_subs}, we re-write Eq. \eqref{eq:w_eq} as
\begin{equation}
  \dfrac{\mathrm{d}u_{zz}}{\mathrm{d}u}=-\frac{2p^{2}}{K}+pv(2+Kp)u^{Kp}+(2Kp+1)\frac{u_{zz}}{u}\,,
  \label{eq:w_eq_new}
\end{equation}
which gives us the following \textit{first integral}:
\begin{equation}  w{\left(z,u,u_{z},u_{zz}\right)}=\frac{1}{u^{(2Kp+1)}}\left(u_{zz}+\frac{(2+Kp)}{K}u_{z}+\frac{(1+Kp)}{K^{2}}u\right)\,.
  \label{eq:w_sol}
\end{equation}
Importantly, we note that for the Lie Algebra $\mathcal{L}=\mathrm{Span}\{X_{z},X_{G}\}$, the general class is independent of $z$, and thus the class is given by $w=F_{zG}(v)$ for an arbitrary function $F_{zG}\in\mathcal{C}(\mathbb{R})$.

% -----------------------------------------------------------------------
% Canonical coordinates
% -----------------------------------------------------------------------
\section{Canonical coordinates of the infinitesimal generators}
Next, we aim at finding first integrals or analytical solutions of the class of models in Theorem 1. To this end, we must carry out two successive step-wise integrations. Consequently, we define two sets of canonical coordinates: one set for the generator of the ideal $\mathcal{L}_{G}=\mathrm{Span}\{X_{G}\}$ and another set for the generator $X_z$ reduced to the invariants of $X_G$.

% -----------------------------------------------------------------------
% Sub-algebra
% -----------------------------------------------------------------------
\subsection{Canonical coordinates for the generalised generator}
Importantly, the first order invariant $v$ in Eq. \eqref{eq:v_sol} in terms of the canonical coordinates $(s_{2},r_{2})$ in Eq. \eqref{eq:canonical_coordinates_2} is given by
  \begin{equation} v=-\left(\frac{1}{Kr_{2}^{Kp+1}}\right)\frac{\mathrm{d}r_{2}}{\mathrm{d}s_{2}}\,,
    \label{eq:v_canonical}
  \end{equation}
and similarly the second order invariant $w$ in Eq. \eqref{eq:w_sol} is given by
\begin{equation}
    w =\frac{1}{K^{2}r_{2}^{(2Kp+1)}}\frac{\mathrm{d}^{2}r_{2}}{\mathrm{d}s_{2}^{2}}\,.
    \label{eq:w_canonical}
\end{equation}
Importantly, we combine Eqs. \eqref{eq:v_canonical} and \eqref{eq:w_canonical} in order to obtain an equation relating $w$ to $v$:
  \begin{equation} w=(Kp+1)v^{2}-\frac{1}{Kr_{2}^{Kp}}\frac{\mathrm{d}v}{\mathrm{d}s_{2}}\,.
    \label{eq:w_and_v}
   \end{equation}
For all the details behind these calculations, see Section \ref{ssec:detailed_sub} of the supplementary material. Subsequently, we calculate the canonical coordinates of the full Lie algebra. 
% -----------------------------------------------------------------------
% Full Lie algebra
% -----------------------------------------------------------------------
\subsection{Canonical coordinates for the reduction of the translation generator}
To calculate the canonical coordinates for the reduction of $X_z$, we begin by considering the zeroth and first order invariants of $X_{G}$ given by $r_{2}$ in Eq. \eqref{eq:canonical_coordinates_2} and $v$ in Eq.~\eqref{eq:v_sol}, respectively. We note that $X_{z}^{(1)}=X_{z}$ and the action of this generator on the two previously mentioned differential invariants is
\begin{align}
  X_{z}r_{2}&=\frac{1}{K}r_{2}\,,\label{eq:X_z_r}\\
  X_{z}^{(1)}v&=0\,.\label{eq:X_z_v}
\end{align}
Therefore, the generator $X_{z}$ restricted to the coordinates $(r_{2},v)$ denoted by $\tilde{X_{z}}$ is given by
\begin{equation} \tilde{X_{z}}=\left.X_{z}\right|_{(r_{2},v)}=\left(\frac{1}{K}\right)r_{2}\partial_{r_{2}}\,.
  \label{eq:X_tilde}
\end{equation}
Based on this generator, the canonical coordinates for the full Lie algebra $\mathcal{L}=\mathrm{Span}\{X_{z},X_{G}\}$ are given by
\begin{equation}
  s_{1}(r_{2},v)=K\ln(r_{2})\,,\quad r_{1}(r_{2},v)=v\,.
  \label{eq:canonical_full}
\end{equation}
Using these canonical coordinates, we conduct a coordinate change from $(s_{2},r_{2})$ to $(s_{1},r_{1})$ in Eq. \eqref{eq:w_and_v} which yields
\begin{equation}
w(s_{1},r_{1})=(Kp+1)r_{1}^{2}+Kr_{1}\frac{\mathrm{d}r_{1}}{\mathrm{d}s_{1}}\,.
  \label{eq:w_ODE}
\end{equation}
The details behind these calculations are presented in Section \ref{ssec:detailed_full}.

%======================================================================
% Action of the symmetries
%======================================================================
\section{Proof of Theorem 1: deriving analytical solutions}
The class of models under consideration is given by $w=F_{zG}(r_{1})$ where $w$ is given by Eq. \eqref{eq:w_ODE}. This can equivalently be formulated as the following first order ODE:
\begin{equation}
K r_{1}\frac{\mathrm{d}r_{1}}{\mathrm{d}s_{1}}=F_{zG}(r_{1})-(Kp+1)r_{1}^{2}\,.
  \label{eq:ODE_step_1_ver_2}
\end{equation}
We seek solutions of Eq. \eqref{eq:ODE_step_1_ver_2} such that the canonical coordinate $r_{1}$ is constant, i.e. $r_{1}=\omega$ for some constant $\omega\neq 0$. By substituting this ansatz into Eq. \eqref{eq:ODE_step_1_ver_2}, we obtain Eq. \eqref{eq:param_class}, and thus the class of models defined by functions $F_{zG}(\omega)$ such that Eq. \eqref{eq:param_class} has at least one root $\omega\in\mathbb{R}$ has solutions of the form $r_{1}=\omega$.

In terms of the canonical coordinates $(s_{2},r_{2})$, the equation $r_{1}=\omega$ corresponds to the following separable autonomous first order ODE
\begin{equation}
\frac{\mathrm{d}s_{2}}{\mathrm{d}r_{2}}=-\frac{1}{K\omega}\frac{1}{r_{2}^{Kp+1}}\,,
    \label{eq:ODE_canonical}
\end{equation}
which gives
\begin{equation}
r_{2}^{Kp}=\left(\frac{1}{K\omega(Kp)}\right)\left(\frac{1}{s_{2}+\tilde{\Upsilon}}\right)\,,
    \label{eq:solution_1}
\end{equation}
 for some integration constant $\tilde{\Upsilon}$. This equation yields
\begin{equation}
r_{2}=\left(\frac{1}{(K\omega)^{1/Kp}(Kp)^{1/Kp}}\right)\left(\frac{1}{(s_{2}+\tilde{\Upsilon})^{1/Kp}}\right)\,.
    \label{eq:solution_2}
\end{equation}
Transforming from $(s_{2},r_{2})$ to the original coordinates $(z,u)$ using the transformations in Eq. \eqref{eq:canonical_coordinates_2}, and then solving the resulting equation for $u$ yields the analytical solution in Eq. \eqref{eq:analytical_solution} where $\Upsilon=\tilde{\Upsilon}/(Kp)$ is an arbitrary integration constant.

%======================================================================
% Action of the symmetries
%======================================================================
\section{Action of the symmetries on analytical solutions}
Here, we generate the symmetries corresponding to the Lie algebra $\mathcal{L}=\mathrm{Span}\{X_{z},X_{G}\}$. Thereafter, we quantify the action of these symmetries on the analytical solutions in Eq. \eqref{eq:analytical_solution}.

\subsection{Generate symmetries from the infinitesimal generators of the Lie group}
Let $\Gamma_{\epsilon}:(z,u)\mapsto(\hat{z}(\epsilon))$ be a (one parameter pointwise) Lie symmetry that is generated by the vector field $X=\xi(z,u)\partial_{z}+\eta(z,u)\partial_u$. Then, the target functions $\hat{t}(\epsilon)$ and $\hat{u}(\epsilon)$ solve the following system of ODEs:
\begin{equation}
  \begin{split}
    \frac{\mathrm{d}\hat{z}}{\mathrm{d}\epsilon}=\xi(\hat{z},\hat{u})\,,\quad\hat{z}(0)=z\,,\\
    \frac{\mathrm{d}\hat{u}}{\mathrm{d}\epsilon}=\eta(\hat{z},\hat{u})\,,\quad\hat{u}(0)=u\,.\\
  \end{split}
  \label{eq:general_ODE_sys}
\end{equation}
Starting with the translation generator $X=\partial_{z}$, it is clear that the corresponding symmetry is given by $\Gamma_{\epsilon}^{z}:(z,u)\mapsto(z+\epsilon,u)$.

For the generator $X_{G}=-K\exp\left(pz\right)\partial_{z}+\exp\left(pz\right)u\partial_{u}$, the target functions solve the following system of ODEs:
\begin{align}
  \frac{\mathrm{d}\hat{z}}{\mathrm{d}\epsilon}&=-K\exp\left(p\hat{z}\right),\quad\hat{z}(0)=z\,,\label{eq:z_hat}\\
  \frac{\mathrm{d}\hat{u}}{\mathrm{d}\epsilon}&=\exp\left(p\hat{z}\right)\hat{u},\quad\hat{u}(0)=u\,.\label{eq:u_hat}
\end{align}
The solution of Eq. \eqref{eq:z_hat} is given by
\begin{align*}
  \frac{\mathrm{d}\hat{z}}{\mathrm{d}\epsilon}&=-K\exp\left(p\hat{z}\right)\\
  \Longrightarrow \int^{\hat{z}}_{z}\exp(-p\hat{z})\mathrm{d}\hat{z}&=-K\int^{\epsilon}_{0}\mathrm{d}\epsilon
  \\
  \Longrightarrow \exp(-p\hat{z})=\exp(-pz)+pK\epsilon&=\frac{pK\epsilon\exp(pz)+1}{\exp(pz)}\\
  \Longrightarrow \exp(p\hat{z})&=\frac{\exp(pz)}{1+pK\epsilon\exp(pz)}\\
  \Longrightarrow \hat{z}&=\frac{1}{p}\ln\left(\frac{\exp(pz)}{1+pK\epsilon\exp(pz)}\right)\,,\\
\end{align*}
which is summarised as follows:
\begin{equation}
\hat{z}(\epsilon)=\frac{1}{p}\ln\left(\frac{\exp(pz)}{1+pK\epsilon\exp(pz)}\right)\,.
  \label{eq:z_trans_2}
\end{equation}
For the second ODE in Eq. \eqref{eq:u_hat}, we note that
\begin{equation}
  \exp(p\hat{z})=\frac{\exp(pz)}{1+pK\epsilon\exp(pz)}\,,
  \label{eq:help_equation}
\end{equation}
which gives
\begin{align*}
  \frac{\mathrm{d}\hat{u}}{\mathrm{d}\epsilon}&=\exp\left(p\hat{z}\right)\hat{u}\\
  \Longrightarrow \int^{\hat{u}}_{u}\frac{1}{\hat{u}}\,\mathrm{d}\hat{u}&=\exp(pz)\int^{\epsilon}_{0}\frac{1}{1+pK\epsilon\exp(pz)}\,\mathrm{d}\epsilon\\
  \Longrightarrow\left[\ln(\hat{u})\right]^{\hat{u}}_{u}&=\frac{1}{pK}[\ln(1+pK\exp(pz)\epsilon)]^{\epsilon}_{0}\\
  \Longrightarrow \ln\left(\frac{\hat{u}}{u}\right)&=\frac{1}{pK}\ln(1+pK\exp(pz)\epsilon)=\ln(\left(1+pK\exp(pz)\epsilon\right)^{1/pK})\\
   \Longrightarrow \hat{u}&=u\left(1+pK\exp(pz)\epsilon\right)^{1/pK}\,,                 
\end{align*}
which is summarised as follows
\begin{equation}
  \hat{u}(\epsilon)=u\left(1+pK\exp(pz)\epsilon\right)^{1/pK}\,.\\  
  \label{eq:u_trans_2}
\end{equation}

\subsection{Action of symmetries on analytical solutions}
We derive the action of the symmetries $\Gamma^{z}_{\epsilon}$ and $\Gamma^{G}_{\epsilon}$, respectively, on the analytical solutions in Eq. \eqref{eq:analytical_solution}. Starting with the translation symmetry, the target functions are given by $\hat{t}(\epsilon)=t+\epsilon$ and $\hat{u}(\epsilon)=u$. The transformed analytical solution curve is therefore described by
\begin{equation}
\hat{u}(\epsilon)=\left(\frac{1}{(K\omega)^{1/Kp}}\right)\left(\frac{1}{\left(1+\hat{\Upsilon}\exp(p\hat{z}(\epsilon))\right)^{1/Kp}}\right)\,,
  \label{eq:u_hat_trans}
\end{equation}
where $\hat{\Upsilon}$ is the parameter defining this transformed solution curve. Now, we wish to derive an equation for $\hat{\Upsilon}$ in terms of $\Upsilon$ and $\epsilon$. By substituting the target functions into Eq. \eqref{eq:u_hat_trans}, we have
\begin{equation}
u=\left(\frac{1}{(K\omega)^{1/Kp}}\right)\left(\frac{1}{\left(1+[\hat{\Upsilon}\exp(p\epsilon)]\exp(pz(\epsilon))\right)^{1/Kp}}\right)\,,
  \label{eq:u_hat_trans_2}
\end{equation}
and by comparing Eqs. \eqref{eq:analytical_solution} and \eqref{eq:u_hat_trans_2} it follows that
\begin{equation}
\Upsilon=\hat{\Upsilon}\exp\left(p\epsilon\right)\,,
  \label{eq:Fisher_param_first}
\end{equation}
which gives
\begin{equation}
\hat{\Upsilon}=\Upsilon\exp\left(-p\epsilon\right)\,.
  \label{eq:Fisher_param_final}
\end{equation}
Moving on to the generalised symmetry $\Gamma_{\epsilon}^{G}$, by substituting the target functions $\hat{z}$ in Eq. \eqref{eq:help_equation} and $\hat{u}$ in Eq. \eqref{eq:u_trans_2}, respectively, into Eq. \eqref{eq:u_hat_trans}, we obtain
\begin{align*}
  u\left(1+pK\exp(pz)\epsilon\right)^{1/pK}&=\left(\frac{1}{(K\omega)^{1/Kp}}\right)\dfrac{1}{\left(1+\hat{\Upsilon}\left(\frac{\exp(pz)}{1+pK\epsilon\exp(pz)}\right)\right)^{1/pK}}\\
  \Longrightarrow u&= \left(\frac{1}{(K\omega)^{1/Kp}}\right)\left(\frac{1}{\left(1+[\hat{\Upsilon}+pK\epsilon]\exp(pz)\right)^{1/Kp}}\right)\,, 
\end{align*}
which is summarised as follows
\begin{equation}
u = \left(\frac{1}{(K\omega)^{1/Kp}}\right)\left(\frac{1}{\left(1+[\hat{\Upsilon}+pK\epsilon]\exp(pz)\right)^{1/Kp}}\right)\,.
  \label{eq:u_hat_trans_3}
\end{equation}
By comparing Eqs. \eqref{eq:analytical_solution} and \eqref{eq:u_hat_trans_3}, we obtain
\begin{equation}
\Upsilon=\hat{\Upsilon}+pK\epsilon\,,
  \label{eq:Fisher_param_third}
\end{equation}
which implies that the parameter of the transformed solution curve is defined by
\begin{equation}
\hat{\Upsilon}=\Upsilon-pK\epsilon\,.
  \label{eq:Fisher_param_very_final}
\end{equation}

\begin{figure}[ht!]
  \centering
  \includegraphics[width=\textwidth]{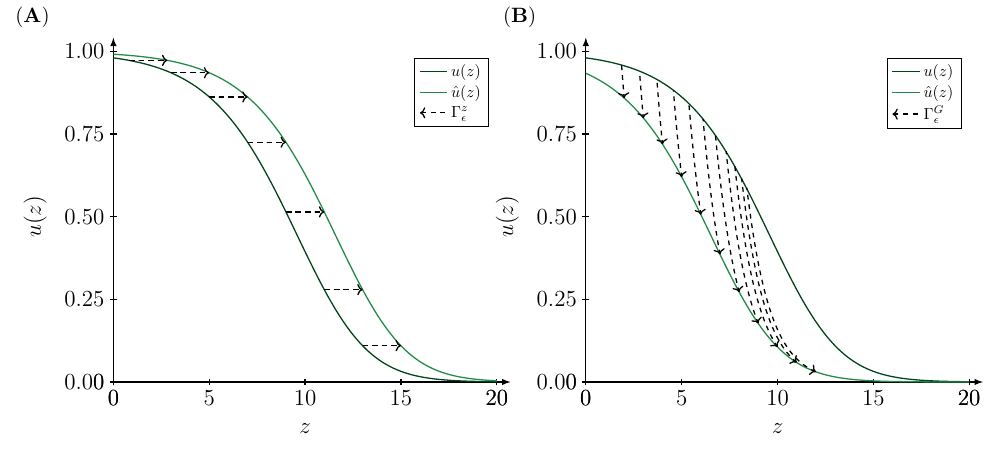}
  \caption{\textit{Action of symmetries on solution curves of the Fisher--KPP model}. The dark green solution curve defined by $\Upsilon=0.01$ of the Fisher--KPP model given by Eq. \eqref{eq:analytical_solution} is transformed to another solution curve in light green where the transformation parameter of the symmetry is denoted by $\epsilon$. This is illustrated in two cases: (\textbf{A}) the action of the translation symmetry $\Gamma_{\epsilon}^{z}$ in Eq. \eqref{eq:translation_symmetry} with a transformation parameter of $\epsilon=2.00$, and (\textbf{B}) the action of the generalised symmetry $\Gamma_{\epsilon}^{G}$ in Eq. \eqref{eq:generalised_symmetry} with a transformation parameter of $\epsilon=0.05$. }
\label{fig:Fisher}
\end{figure}

In summary, the infinitesimal generator of the Lie group $X_{z}=\partial_z$ generates the symmetry $\Gamma_{\epsilon}^{z}$ given by
\begin{equation}
  \Gamma_{\epsilon}^{z}:(z,u)\mapsto(z+\epsilon,u)\,,
  \label{eq:translation_symmetry}
\end{equation}
corresponding to translations in $z$. Moreover, the vector field $X_{G}$ in Eq. \eqref{eq:generalised_generator} generates the symmetry $\Gamma_{\epsilon}^{G}$ given by Eq. \eqref{eq:generalised_symmetry}.

To illustrate the transformations of solution curves by these symmetries, we consider the Fisher--KPP model defined by
\begin{equation}
    u_{zz}+\frac{5}{\sqrt{6}}u_{z}+u(1-u)=0\,,
    \label{eq:Fisher_KPP}
\end{equation}
 which is given by the parameters
\begin{equation}
  p=\frac{1}{\sqrt{6}}\,,\quad K=\sqrt{\frac{3}{2}}\,.
  \label{eq:Fisher_param}
\end{equation}
Based on this example, the action of the symmetries $\Gamma_{\epsilon}^{z}$ and $\Gamma_{\epsilon}^{G}$ in Eqs. \eqref{eq:translation_symmetry} and \eqref{eq:generalised_symmetry}, respectively, on the analytical solutions in Eq. \eqref{eq:analytical_solution} is illustrated in Figure \ref{fig:Fisher}.

%======================================================================
% Canonical coordinates
%======================================================================
\section{Detailed calculations of the canonical coordinates}

Here, we calculate the canonical coordinates of the two generators of the Lie algebra $\mathcal{L}=\mathrm{Span}\{X_{z},X_{G}\}$.

%======================================================================
% Full Lie algebra
%======================================================================
\subsection{Canonical coordinates of the generalised generator}\label{ssec:detailed_sub}

We consider the following differential invariants:
  \begin{align}    v(u,u_{z})&=\frac{1}{u^{Kp}}\left(\frac{u_{z}}{u}+\frac{1}{K}\right)\,,\label{eq:v}\\ w(u,u_{z},u_{zz})&=\frac{1}{u^{(2Kp+1)}}\left(u_{zz}+\left(\frac{2+Kp}{K}\right)u_{z}+\left(\frac{1+Kp}{K^{2}}\right)u\right)\,,\label{eq:w}    
  \end{align}
together with the following infinitesimal generator of the Lie group $X_{G}$ in Eq. \eqref{eq:generalised_generator}. In other words, the infinitesimals are given by $\xi{(z)}=-K\exp\left(pz\right)$ and $\eta{(z,u)}=\exp\left(pz\right)u$. The canonical coordinate $r_{2}{(z,u)}$ is given by the zeroth first integral in Eq. \eqref{eq:r_sol}. The canonical coordinate $s_{2}(z)$ is defined by $s_{2}(z)=\int\xi(z)^{-1}\mathrm{d}z$ and is given by
   \begin{equation}
    s_{2}{(z)}=\frac{1}{Kp}\exp(-pz)\,.
    \label{eq:s_sol}
\end{equation}
By combining Eqs. \eqref{eq:r_sol} and \eqref{eq:s_sol}, we derive the following two non-linear coordinate transformations for $z$ and $u$ in terms of these canonical coordinates
\begin{align}
    z{(s_{2})}&=-\frac{1}{p}\ln(Kps_{2})\,,\label{eq:z_trans}\\
    u{(s_{2})}&=r_{2}(s_{2})(Kps_{2})^{1/Kp}\,.\label{eq:u_trans}
\end{align}
By differentiating Eq. \eqref{eq:z_trans} with respect to $s_{2}$, we obtain
\begin{equation} \frac{\mathrm{d}z}{\mathrm{d}s_{2}}=-\frac{1}{ps_{2}}\Leftrightarrow\frac{\mathrm{d}s_{2}}{\mathrm{d}z}=-ps_{2}=-\left(\frac{1}{K}\right)(pKs_{2})\,,
    \label{eq:dzds}
\end{equation}
and by differentiating Eq. \eqref{eq:u_trans} with respect to $s_{2}$, we obtain
\begin{equation} \frac{\mathrm{d}u}{\mathrm{d}s_{2}}=r_{2}'(s_{2})(Kps_{2})^{1/Kp}+r_{2}(s_{2})(Kps_{2})^{(1/Kp)-1}\,.
    \label{eq:duds}
\end{equation}
Since $u_{z}=\mathrm{d}u/\mathrm{d}z=(\mathrm{d}u/\mathrm{d}s_{2})(\mathrm{d}s_{2}/\mathrm{d}z)$, we have
\begin{equation} u_{z}=-\left(\frac{1}{K}\right)\left(r_{2}'(s_{2})(Kps_{2})^{(1/Kp)+1}+r_{2}(s_{2})(Kps_{2})^{1/Kp}\right)\,.
    \label{eq:uz}
\end{equation}
Differentiating $u_{z}$ with respect to $s_{2}$, yields
\begin{equation}
    \frac{\mathrm{d}u_{z}}{\mathrm{d}s_{2}}=-\left(\frac{1}{K}\right)\left(r_{2}''(s_{2})(Kps_{2})^{(1/Kp)+1}+(2+Kp)r_{2}'(s_{2})(Kps_{2})^{1/Kp}+r_{2}(s_{2})(Kps_{2})^{(1/Kp)-1}\right)\,,
    \label{eq:duzds}
\end{equation}
and since $u_{zz}=\mathrm{d}u_{z}/\mathrm{d}z=(\mathrm{d}u_{z}/\mathrm{d}s_{2})(\mathrm{d}s_{2}/\mathrm{d}z)$, we have
\begin{equation} u_{zz}=\frac{1}{K^{2}}\left(r_{2}''(s_{2})(Kps_{2})^{(1/Kp)+2}+(2+Kp)r_{2}'(s_{2})(Kps_{2})^{(1/Kp)+1}+r_{2}(s_{2})(Kps_{2})^{1/Kp}\right)\,.
    \label{eq:uzz}    
  \end{equation}
  Using these expressions for the derivatives in terms of the canonical coordinates, we can express the differential invariants in terms of the canonical coordinates. Beginning with $v$ in Eq. \eqref{eq:v}, substituting Eqs. \eqref{eq:u_trans} and \eqref{eq:uz} into this equation yields
  \begin{equation}
  v=-\left(\frac{1}{Kr_{2}^{Kp+1}}\right)\frac{\mathrm{d}r_{2}}{\mathrm{d}s_{2}}\,.
    \label{eq:v_canonical_2}
  \end{equation}
Next, from Eqs. \eqref{eq:u_trans}, \eqref{eq:uz} and \eqref{eq:uzz} we have
\begin{equation} u_{zz}+\left(\frac{2+Kp}{K}\right)u_{z}+\left(\frac{1+Kp}{K^{2}}\right)u=\frac{1}{K^{2}}r_{2}''(s_{2})(Kps_{2})^{(1/Kp)+2}\,.
    \label{eq:LHS}
\end{equation}
Also, from Eq. \eqref{eq:u_trans} it follows that
\begin{equation}
    u^{(2Kp+1)}=r_{2}(s_{2})^{(2Kp+1)}(Kps_{2})^{(1/Kp)+2}\,,
    \label{eq:RHS}
\end{equation}
and dividing Eq. \eqref{eq:LHS} by Eq. \eqref{eq:RHS} yields
\begin{equation}
    w = \frac{1}{K^{2}r_{2}(s_{2})^{(2Kp+1)}}r_{2}''(s_{2})\,,
    \label{eq:w_canonical_2}
\end{equation}
according to Eq. \eqref{eq:w}.

Lastly, we calculate an equation linking $v$ in Eq. \eqref{eq:v_canonical_2} to $w$ in Eq. \eqref{eq:w_canonical_2}. From Eq. \eqref{eq:v_canonical_2}, we have that
  \begin{equation}
    \frac{\mathrm{d}r_{2}}{\mathrm{d}s_{2}}=-vKr_{2}^{Kp+1}\,,
    \label{eq:v_canonical_1}
  \end{equation}
and thus
  \begin{equation}  \frac{\mathrm{d}^{2}r_{2}}{\mathrm{d}s_{2}^{2}}=-\frac{\mathrm{d}v}{\mathrm{d}s_{2}}Kr_{2}^{Kp+1}-vK(Kp+1)r_{2}^{Kp}\frac{\mathrm{d}r_{2}}{\mathrm{d}s_{2}}\,.
    \label{eq:v_canonical_3}
  \end{equation}
 Moreover, Eq. \eqref{eq:w_canonical_2} yields
\begin{equation}
 \frac{\mathrm{d}^{2}r_{2}}{\mathrm{d}s_{2}^{2}}=K^{2}r_{2}^{(2Kp+1)}w \,,
    \label{eq:w_canonical_1}
  \end{equation}
  and by equating Eqs. \eqref{eq:v_canonical_3} and \eqref{eq:w_canonical_1} and then solving for $w$ results in
  \begin{equation} 
  w=(Kp+1)v^{2}-\frac{1}{Kr_{2}^{Kp}}\frac{\mathrm{d}v}{\mathrm{d}s_{2}}\,.
    \label{eq:w_and_v_2}
   \end{equation}

%======================================================================
% Full Lie algebra
%======================================================================
   \subsection{Canonical coordinates of the translation generator}\label{ssec:detailed_full}
   Consider the generator $X_{z}$ restricted to the coordinates $(r_{2},v)$ denoted by $\tilde{X_{z}}$ in Eq. \eqref{eq:X_tilde}. For this generator, we have the following two infinitesimals:
\begin{equation}
  \tilde{\xi}(r_{2},v)=\frac{1}{K}r_{2}\,,\quad \tilde{\eta}(r_{2},v)=0\,.
  \label{eq:infinitesimals_tilde}
\end{equation}
The canonical coordinate $s_{1}$ is given by
\begin{equation*}
  s_{1}=K\int\frac{1}{r_{2}}\mathrm{d}r_{2}=K\ln(r_{2})\,.
\end{equation*}
The canonical coordinate $r_{1}$ is a first integral of
\begin{equation}
\frac{\mathrm{d}v}{\mathrm{d}r_{2}}=0\,,
\end{equation}
and hence $r_{1}=v$. In summary, the canonical coordinates for the full Lie algebra are given by Eq. \eqref{eq:canonical_full}. Next, we want express Eq.~\eqref{eq:w_and_v} in terms of the canonical coordinates $(s_{1},r_{1})$ of the reduced generator $\tilde{X_z}$ instead of the canonical coordinates $(s_{2},r_{2})$ of $X_{G}$. To this end, we need to change coordinates of the term
\begin{equation}
  \frac{1}{Kr_{2}^{Kp}}\frac{\mathrm{d}v}{\mathrm{d}s_{2}}=\frac{1}{Kr_{2}^{Kp}}\frac{\mathrm{d}r_{1}}{\mathrm{d}s_{2}}\,.
  \label{eq:lonely_term}
\end{equation}
To this end, we have that
\begin{equation}
  \frac{\mathrm{d}r_{1}}{\mathrm{d}s_{2}}=\frac{\mathrm{d}r_{1}}{\mathrm{d}s_{1}}\frac{\mathrm{d}s_{1}}{\mathrm{d}s_{2}}=\frac{\mathrm{d}r_{1}}{\mathrm{d}s_{1}}\left(\dfrac{\frac{\mathrm{d}s_{1}}{\mathrm{d}z}}{\frac{\mathrm{d}s_{2}}{\mathrm{d}z}}\right)\,.
  \label{eq:derivative}
\end{equation}
Moreover, it follows that:
\begin{equation}
  \mathrm{d}s_{1}=\frac{\mathrm{d}s_{1}}{\mathrm{d}z}=K\left(\frac{u_{z}}{u}+\frac{1}{K}\right)\,,
  \label{eq:ds1}
\end{equation}
and
\begin{equation}
  \mathrm{d}s_{2}=\frac{\mathrm{d}s_{2}}{\mathrm{d}z}=-\frac{1}{K}\exp(-pz)\,.
  \label{eq:ds2}
\end{equation}
Now, substituting Eqs. \eqref{eq:ds1} and \eqref{eq:ds2} into Eq. \eqref{eq:derivative} and then substituting the resulting expression into Eq. \eqref{eq:lonely_term} results in the following calculations
\begin{align*}
  \frac{1}{Kr_{2}^{Kp}}\frac{\mathrm{d}r_{1}}{\mathrm{d}s_{2}}&=-\frac{1}{K}\frac{\mathrm{d}r_{1}}{\mathrm{d}s_{1}}K\underset{=r_{1}}{\underbrace{\left[\frac{1}{u^{Kp}}\left(\frac{u_{z}}{u}+\frac{1}{K}\right)\right]}}K\underset{=1}{\underbrace{\left(\dfrac{\exp(pz)}{\exp\left(\frac{z}{K}\right)^{Kp}}\right)}}\\
&=-Kr_{1}\frac{\mathrm{d}r_{1}}{\mathrm{d}s_{1}}\,.
\end{align*}
In summary, Eq. \eqref{eq:lonely_term} can be written as follows:
\begin{equation}
\frac{1}{Kr_{2}^{Kp}}\frac{\mathrm{d}r_{1}}{\mathrm{d}s_{2}}=-Kr_{1}\frac{\mathrm{d}r_{1}}{\mathrm{d}s_{1}}\,,
  \label{eq:lonely_term_updated}
\end{equation}
which implies that the equation for $w$ in terms of the canonical coordinates $(s_{1},r_{1})$ is given by Eq. \eqref{eq:w_ODE}.

\end{document}